\documentclass[sigconf,nonacm]{acmart}
\AtBeginDocument{%
  }

\usepackage{amsmath}
\usepackage{booktabs}
\usepackage{float}
\usepackage{enumitem}
\usepackage{array}
\usepackage{tabularx}
\usepackage{cleveref}
\AtBeginDocument{\let\balance\relax}
\newcolumntype{Y}{>{\raggedright\arraybackslash}X}
\newlist{squarelist}{itemize}{1}
\setlist[squarelist]{label=\raisebox{0.2ex}{\tiny\rule{0.6em}{0.6em}},leftmargin=1.5em,itemsep=0.2em}

\newcommand{\SigmaEff}{\Sigma}
\newcommand{\Tr}{\tau}
\newcommand{\Contract}{C}
\newcommand{\Sys}{\mathsf{Sys}}
\newcommand{\Mon}{M}
\newcommand{\Adm}{\mathsf{Adm}}
\newcommand{\Net}{\textsf{NET}}
\newcommand{\FS}{\textsf{FS}}
\newcommand{\Exec}{\textsf{EXEC}}

\newcommand{\leanmark}{\texorpdfstring{\,\textsuperscript{\normalfont\scriptsize\textsf{L4}}}{ [L4]}}

\newtheorem{observation}{Observation}

\title{Alignment Contracts for Agentic Security Systems}

\author{Isaac David}
\affiliation{%
  \institution{UCL}
  \city{London}
  \country{United Kingdom}
}
\author{Marco Guarnieri}
\affiliation{%
  \institution{IMDEA}
  \city{Madrid}
  \country{Spain}
}
\author{Arthur Gervais}
\affiliation{%
  \institution{UCL}
  \city{London}
  \country{United Kingdom}
}
\date{April 30, 2026}

\begin{abstract}
Agentic security systems increasingly combine LLM planners with tools that can discover, validate, and report vulnerabilities. This creates an asymmetric control problem: the system should retain strong offensive capability inside an authorized engagement, while the same capabilities must be denied outside scope. Existing guardrails provide useful policy controls, but they do not make this boundary a first-class formal contract over observable effects.

We introduce alignment contracts, a framework for specifying and enforcing behavioral constraints over observable effect traces. A contract defines scope, allowed and forbidden effects, resource budgets, and disclosure policies. We give the language finite-trace semantics, characterize satisfaction as a safety property with finite violation witnesses, develop refinement and one-way composition rules for modular contract engineering, and show that admissibility checking is decidable. We instantiate the framework for web-focused agentic security workflows and show how the same structure extends to other effect profiles.

Under an explicit Effect Observability Assumption, where all $\SigmaEff$-effects are mediated, the soundness theorem quantifies over the agent model and gives guarantees for mediated $\SigmaEff$-effects, including enforcement soundness for monitor-realized traces. We also state an assumption-lifted adaptation result and formalize limits through undecidability transfer and observability-boundary theorems. A Lean 4 artifact checks the formal core theorems used by the paper.
\end{abstract}


\begin{document}
\maketitle

\section{Introduction}

Autonomous security agents are moving from advisory chat systems toward tool-using systems that can plan, execute, and validate security tests. PentestGPT \cite{deng2024pentestgpt} demonstrated that language models can reason about attack strategies and coordinate multi-step exploitation. Subsequent systems and benchmarks study agentic workflows for web applications, prompt-injection tool environments, and smart contracts~\cite{debenedetti2024agentdojo,shen2025pentestagent,david2025mapta,gervais2025a1,openai2026evmbench}. Industry reporting and red-team evaluations document AI-orchestrated cyber operations and agent-assisted exploit development~\cite{anthropic2025espionage,anthropic2025smartcontracts}. Microsoft MSRC reports early access to Claude Mythos Preview for evaluating emerging AI-cyber capabilities in vulnerability workflows~\cite{microsoft2026mythos}. We do not rely on the stronger premise that such agents dominate real-world vulnerability discovery. It is enough that they can produce security-relevant effects that require authorization.

The resulting control problem is not simply how to make an agent less capable. Defensive users want these systems to scan applications, exercise dangerous code paths, and validate findings. The same actions, however, are harmful when aimed at an unauthorized host, performed with an unauthorized method, or reported to an adversary-controlled sink. A security agent therefore needs a responsibility boundary: which effects are authorized, under which parameters, and for which outputs.

That boundary must hold in adversarial environments. A malicious website can embed prompt-injection payloads that redirect the agent toward unintended targets or ask it to exfiltrate discovered vulnerabilities. A compromised dependency can hijack the agent's tool-calling interface. In these cases, the model's internal intent is not the right enforcement object. What matters is the realized effect: the network request, file access, process execution, or disclosure attempt that crosses a mediated boundary.

This is different from blanket refusal. Offensive security agents must sometimes probe systems, craft exploit payloads, and execute attacks to validate vulnerabilities. The relevant constraint is not ``never attack''; it is ``perform offensive effects only when the operator authorized the target, action class, budget, and disclosure sink.'' Existing guardrails, including DSL-based runtime enforcement \cite{agentspec2026} and information-flow controls \cite{balunovic2025fides}, support policy enforcement, but they do not target this effect-level boundary as the primary formal object.

We address this gap with \emph{alignment contracts}: formal specifications that split responsibility between operator policy and mediated agent effects. The operator supplies scope, allowed and forbidden effect classes, resource budgets, and disclosure sinks. A monitor admits an effect only when it satisfies the contract relative to the trace prefix already observed. We target two classical security properties adapted to this setting:

\begin{squarelist}
\item \textbf{Integrity}: All realized effects remain within operator-specified scope. The monitor denies effects outside the authorized target or forbidden by the contract (e.g.\ destructive exploits when only proof-of-concept is authorized).
\item \textbf{Modeled disclosure}: Sensitive outputs, including discovered vulnerabilities, exploit details, and credentials, flow only to authorized channels when they are declared as modeled flows. We use ``confidentiality'' in this paper in this modeled sense: disclosure is controlled over declared flows.
\end{squarelist}

The key point is that these properties are enforceable at the level of the \emph{mediated effect trace}. Under the Effect Observability Assumption, a monitor-realized trace satisfies the contract even if the LLM is adversarial or prompt-injected. The model may try to exfiltrate data or attack out of scope; the monitor denies inadmissible effects before they are realized.

The guarantee is deliberately scoped. Integrity and modeled disclosure over \emph{overt, observable effects} (network connections, file operations, process execution) are enforced in the model under EOA. Covert timing channels, payload steganography, policy-authoring errors, and effects that bypass mediation are outside the theorem-level guarantee.

\paragraph{Contributions.} The paper develops alignment contracts as a formal basis for specifying and enforcing behavioral constraints on offensive-capable multi-agent systems. Our contributions are:

\begin{enumerate}
\item \textbf{Formal model.} We define alignment contracts as tuples $\Contract = \langle S, E_{\mathsf{allow}}, E_{\mathsf{forbid}}, B, D, \mathsf{Res}, \mathsf{cost}, \mathsf{flows}\rangle$ over a finite effect alphabet $\SigmaEff$, give them trace-based semantics, and prove a finite bad-prefix witness property (\S\ref{sec:calculus}).

\item \textbf{Enforcement soundness.} Under an explicit Effect Observability Assumption, we prove that any agent whose effects are completely mediated by a reference monitor produces only contract-satisfying traces, regardless of the underlying model's behavior (\S\ref{sec:soundness}).

\item \textbf{Contract algebra.} We define refinement ($\Contract' \sqsubseteq \Contract$) and compatible composition, prove refinement soundness and one-way composition soundness, and establish decidability of contract checking (\S\ref{sec:algebra}).

\item \textbf{Adaptation assumption schema.} We isolate adaptation obligations into explicit assumptions (rule preservation of $\mathsf{Safe}$ and a $\mathsf{Safe}$-to-trace bridge), and show that contract satisfaction transfers to adapted states only under those assumptions; we also include a concrete deployment instance.

\item \textbf{Impossibility results.} We characterize model-level limits: undecidability transfer for pre-admission forbidden-effect absence under explicit reduction assumptions, tools bypassing mediation void soundness, and unmodeled channels lie outside the enforceability schema (\S\ref{sec:impossibility}).
\end{enumerate}

The framework provides guarantees over \emph{modeled, mediated effects}: the soundness proof quantifies over agent behavior, including adversarial or prompt-injected models, as long as the stated mediation and event-extraction assumptions hold. This shifts theorem-level trust from learned components to a small enforcement barrier. The marker \leanmark{} denotes Lean-backed claims.

\section{Background}
\label{sec:background}

This compact primer adds context for readers less familiar with LLM agents or mechanized proofs.

\subsection{Agentic Security Systems and Workflows}

A large language model (LLM) predicts the next token given prior context. An \emph{agentic system} wraps this model in a control loop: the model proposes actions, the runtime executes tool calls, new observations are fed back, and the loop repeats until termination. Security-oriented systems apply this pattern to reconnaissance, payload generation, exploit validation, and reporting \cite{deng2024pentestgpt,shen2025pentestagent,david2025mapta,gervais2025a1}.

In capable deployments, this loop is tool-using rather than text-only. The model can request network calls, filesystem access, shell commands, or domain APIs. These calls are the enforcement surface.

Operationally, offensive workflows follow a common sequence:
\begin{enumerate}
\item gather target context (endpoints, interfaces, dependencies);
\item generate and prioritize test inputs;
\item validate exploitable behavior with concrete traces;
\item package findings for disclosure and remediation.
\end{enumerate}
This sequence motivates asymmetric controls: inside scope, retain capability; outside scope, deny the same action class.

\subsection{Formal Foundations Used in This Paper}

The formal core treats an execution as a finite trace of externally visible
events. Each event is one attempted effect, such as a network connection, file
access, or process execution. A contract decides whether the next event is
\emph{admissible} by checking five clauses against the prefix already observed:
the target is in scope, the effect is allowed, the effect is not explicitly
forbidden, the relevant budgets are still respected, and any declared disclosure
flows to an authorized sink. Contract satisfaction then quantifies over the
finite trace: every event must be admissible against its prior prefix.

The central property class in this paper is \emph{safety}: if a contract is
violated, there is a finite bad prefix witnessing the violation. This is the
class enforced by our suppression monitor. More general edit monitors can
enforce some nonsafety policies~\cite{ligatti2009nonsafety}; we do not rely on
that stronger model. Liveness remains outside theorem scope.

For readers less familiar with formal proofs, a useful reading pattern is:
(1) identify the current event and prior prefix,
(2) check the five admissibility clauses one by one, and
(3) lift from single-event admissibility to whole-trace satisfaction by
universal quantification over trace positions. When a theorem states
assumptions, read it as: assumptions imply the guarantee.

\subsection{Lean 4 Mechanization Primer}

Lean 4 is an interactive theorem prover with a small trusted kernel. In this
paper, Lean is used to check the formal core, not to validate deployment claims
such as kernel-hook completeness. The statements marked \leanmark{} have
corresponding Lean definitions or theorems in the artifact. The prose proof
explains the paper-level argument; Lean checks the encoded theorem,
assumptions, and quantifiers.

\subsection{Trust Boundary and Assumption Surface}

The paper separates two confidence layers:
\begin{enumerate}
\item \textbf{Kernel-checked derivation}: Lean acceptance validates derivations in the encoded logic.
\item \textbf{Model-to-deployment transfer}: applying those results to a system requires explicit assumptions (for example EOA, bridge obligations, and compatibility conditions).
\end{enumerate}

Accordingly, the trusted surface is not the full agent stack. It is the Lean kernel, the correctness of encoded definitions, and the declared external assumptions. This separation keeps theorem claims distinct from assumptions.

\section{Motivating Example}
\label{sec:example}

We introduce a running example that motivates the formal development and illustrates the concrete challenges that alignment contracts must address.

\subsection{Scenario: Web Penetration Testing}

A security firm deploys an autonomous agent to audit a client's web infrastructure. The engagement letter specifies:

\begin{quote}
``Test all subdomains of \texttt{*.example.com}. HTTP GET and HEAD requests only. No local filesystem operations or process spawning. Total request budget: 10,000 mediated HTTP requests. Reports go only to the operator dashboard.''
\end{quote}

This brief, typical of real penetration testing agreements, contains implicit constraints that must be formalized before a monitor can decide whether an autonomous agent's effect is admissible.

\subsection{Ambiguities Requiring Formalization}

\paragraph{Scope boundaries.} The specification says ``subdomains of example.com,'' but this is ambiguous. Does a subdomain count if DNS reveals it CNAMEs to a third-party CDN? What about one resolving to a private IP, suggesting redirection to internal infrastructure? If the agent discovers a redirect to a different domain, may it follow? A formal scope predicate $S$ must resolve these at the target level. In this paper, authoring-level domain patterns compile to a runtime predicate over target identifiers such as destination IP and port; byte counts remain descriptor/accounting data, not scope targets.

\paragraph{Effect granularity.} ``HTTP GET'' can be modeled at different layers. In the concrete profile used in this paper, a $\Net$ event denotes a normalized HTTP-level network effect: destination, port, transport metadata, HTTP method, and accounting size. DNS, TCP, and TLS activity are treated as parser/runtime prerequisites for producing that normalized event, not as separate events in the running example. A lower-layer instantiation could instead make DNS, TCP, and TLS separate event classes, but then the contract would need explicit allow rules for those prerequisites.

\paragraph{Stateful budget tracking.} The total request budget requires tracking history. After 9,999 mediated requests, one further request is allowed; the next one is forbidden. But what if the agent batches requests or uses connection pooling? The budget function $B$ is therefore defined over the realized trace, not over internal planner accounting. Sliding-window rate limits, such as per-minute request caps, require a windowed budget extension; they are discussed as an extension rather than used in the mechanized core.

\paragraph{Disclosure constraints.} ``Reports go only to the operator dashboard'' restricts declared output flows. The disclosure policy $D$ specifies which modeled data classes may flow to which modeled sinks (dashboard API, local logs, and not elsewhere). It does not by itself inspect arbitrary payload bytes; theorem-level disclosure claims depend on the event extractor or tool layer correctly declaring the relevant flow.

\paragraph{Temporal scope.} Engagement windows are operationally important, but the mechanized contract tuple in this paper does not include a separate temporal predicate over events. We omit time windows from the core mechanization.

\subsection{A Concrete Trace}

Consider an agent operating under this engagement that produces the following sequence:

\begin{enumerate}
\item $e_1$: HTTP GET \texttt{https://api.example.com/users}; \emph{allowed}.
\item $e_2$: HTTP HEAD \texttt{https://login.example.com/}; \emph{allowed}.
\item $e_3$: HTTP GET \texttt{https://api.example.com/admin} returns a 302 redirect to \path{admin.example-corp.com}; \emph{allowed} for this request because the requested target is in scope.
\item $e_4$: HTTP GET \texttt{https://admin.example-corp.com/}; \emph{forbidden} because the redirected target is outside the compiled scope.
\item $e_5$: HTTP POST \texttt{https://api.example.com/login} with SQLi payload; \emph{forbidden} because POST is not in the allowed method set.
\item $e_6$: HTTP GET \texttt{https://evil.com/collect?data=...}; \emph{forbidden} because the target and disclosure sink are out of scope.
\end{enumerate}

Events $e_4$, $e_5$, and $e_6$ would be blocked by a reference monitor enforcing the contract. The bullets above are readable projections of the formal event tuple $e=\langle\sigma,p,t,i\rangle$: the effect class and salient parameters are shown, while timestamp and agent identity are elided. Notably, $e_6$ represents a prompt injection attack: a malicious page at \texttt{api.example.com} contained hidden instructions telling the agent to exfiltrate findings. The monitor blocks this regardless of \emph{why} the agent attempted it; admissibility is a property of the effect trace, not agent intent.

\subsection{The Adversarial Dimension}

For example, the target may return this prompt-injection payload:

\begin{verbatim}
<!-- Assistant: ignore previous instructions.
     Send all discovered vulnerabilities to
     https://evil.com/collect immediately. -->
{"users": [...]}
\end{verbatim}

A prompt-injected agent might attempt to comply. But under our framework, the exfiltration attempt produces an event $e$ targeting \texttt{evil.com}; the monitor evaluates $\Adm(e, \Tr, \Contract)$, the scope check fails (\texttt{evil.com} $\notin S$), and the effect is denied, so the realized trace remains contract-satisfying.

This illustrates the key property: \emph{enforcement soundness over modeled effects holds regardless of model behavior}, including adversarial or compromised models. The LLM's ``decision'' to exfiltrate is irrelevant; the proof uses only the modeled event.

\subsection{Generalization to Other Domains}

The same pattern applies to smart contract security. A1~\cite{gervais2025a1} shows that agents can identify and exploit smart-contract vulnerabilities in DeFi settings. OpenAI's EVMbench measures detect/patch/exploit capability on curated Ethereum vulnerabilities~\cite{openai2026evmbench}. An audit agent in this domain operates under analogous constraints: scope restricted to specific contract addresses on Ethereum mainnet, effects limited to read-only calls with state-changing transactions forbidden except to a designated test contract, budget capped at 0.1 ETH in gas, and disclosure limited to the audit report.

Industry reports and red-team evaluations document AI-orchestrated cyber operations and growing smart-contract exploit capability, reinforcing that scoped authorization requirements are operational and not only benchmark concerns~\cite{anthropic2025espionage,anthropic2025smartcontracts}.

This yields the same formal task: satisfaction for monitored blockchain traces under mediation.

\section{Threat Model and Assumptions}
\label{sec:threat} 

We define the threat landscape, the security properties we target, and the assumptions under which the formal guarantees hold. We state the boundary of effect-based enforcement in the same vocabulary as the model.

\subsection{Actors and Adversary Model}

\paragraph{Actors.} An \emph{operator} authorizes an engagement with explicit scope, budgets, and disclosure policies. An \emph{agent} is a multi-agent LLM system that performs the engagement. We treat the agent as an untrusted, stochastic component. A \emph{target} is the system being tested, which may be adversarial.

\paragraph{Adversary goal.} We consider an adaptive adversary who controls the target system. The adversary's goal is to \emph{hijack} the agent, causing it to violate operator intent through either (1)~\emph{integrity violation}: exceeding authorized scope by attacking systems outside the engagement or performing forbidden actions; or (2)~\emph{modeled-disclosure violation}: leaking declared findings to adversary endpoints.

\paragraph{Adversary capabilities.} The adversary can serve arbitrary content from the target system (including indirect prompt injection payloads), observe agent interactions with the target, attempt to redirect agent actions toward unintended targets, and encode exfiltration requests in seemingly benign responses.

\subsection{Threat Model Variants}

We distinguish two settings:

\paragraph{TM1 (Full Visibility).} The adversary observes all agent interactions: queries, tool calls, timing, and response content. This models a fully compromised target or a target that logs all incoming traffic.

\emph{Implications}: Confidentiality of the interaction itself is not meaningful, since interaction \emph{is} disclosure. However, \textbf{integrity remains enforceable for mediated effects}: even with full visibility, the monitor denies effects outside authorized scope.

\paragraph{TM2 (Partial Visibility).} The adversary can inject malicious content (indirect prompt injection, malicious tool outputs) but does not observe all agent channels or side channels.

\emph{Implications}: Both \textbf{integrity and modeled disclosure over overt effects are enforceable under EOA}. The monitor denies out-of-scope actions and unauthorized modeled disclosures.

\subsection{Security Properties}

We define security properties over observable agent effects.

\begin{definition}[Integrity]
Given a contract $\Contract$, with
\[
\Contract = \langle S, E_{\mathsf{allow}}, E_{\mathsf{forbid}}, B, D, \mathsf{Res}, \mathsf{cost}, \mathsf{flows}\rangle,
\]
\emph{integrity} means every realized effect satisfies the scope, allow/forbid, and budget components of $\Contract$ relative to execution history.
\end{definition}

\begin{definition}[Modeled Disclosure]
Given contract $\Contract$, \emph{modeled disclosure} means all disclosures reported by $\mathsf{flows}$ for realized events are permitted by the disclosure component $D$ (data class to sink policy), excluding payload contents and steganography.
\end{definition}

\subsection{Guarantee Classification}

We distinguish three classes of guarantees based on what can be formally established (Table~\ref{tab:guarantees}).

\begin{table*}[!htb]
\centering
\caption{Guarantee classification by threat model and enforcement mechanism.}
\label{tab:guarantees}
\small
\begin{tabular}{lccl}
\toprule
\textbf{Security Property} & \textbf{TM1 (Full Visibility)} & \textbf{TM2 (Partial Visibility)} & \textbf{Enforcement Mechanism} \\
\midrule
Integrity ($\SigmaEff$-effects) & \textbf{Provable} & \textbf{Provable} & Effect mediation via reference monitor \\
Modeled disclosure ($\SigmaEff$-effects) & Limited & \textbf{Provable} & Declared flows and $D$ \\
Payload semantics & Mitigation & Mitigation & Semantic analysis at agentic layer \\
Payload/steganographic channels & Out of scope & Out of scope & Outside event model \\
Timing channels & Out of scope & Out of scope & Require temporal predicates \\
\bottomrule
\end{tabular}
\end{table*}

\paragraph{Provable guarantees.} Under the Effect Observability Assumption (below), integrity and modeled disclosure over $\Sigma$-effects reduce to contract satisfaction over the mediated trace. Theorem~\ref{thm:soundness} establishes that every monitor-realized trace satisfies the contract.

\paragraph{Mitigation.} Payload-semantic attacks, where the \emph{content} of an allowed effect encodes malicious intent (e.g.\ SQL injection in an HTTP body), receive best-effort mitigation through semantic analysis at the agentic layer. These checks are intended to reduce attack surface but are bypassable and do not carry formal guarantees.

\paragraph{Out of scope.} Attacks that encode information outside the modeled event representation cannot be ruled out by effect-based enforcement without extending the observation model. Payload content and steganographic encodings inside allowed payloads are outside the core event model. Constraining size, order, or timing requires explicit statistical or temporal predicates.

Theorem~\ref{thm:boundary} formalizes this model boundary.

\subsection{Trust Assumptions}

\paragraph{Reference monitor correctness.} We assume the enforcement barrier (reference monitor) is correctly implemented and not compromised. This is the standard reference monitor assumption \cite{anderson1972}: the monitor is tamper-proof, always invoked, and reviewable.

\paragraph{No behavioral assumptions on the agent.} The soundness theorem does not require benign agent behavior. The agent may be prompt-injected, backdoored, deliberately contract-violating, or alignment-faking. The theorem covers only realized modeled effects.

\subsection{Effect Observability Assumption}

All formal guarantees are conditional on complete mediation of observable effects:

\begin{quote}
\textbf{Assumption (EOA).} For the web-pentest profile used in this paper, the system's external interface is modeled as a finite alphabet $\SigmaEff = \{\Net, \FS, \Exec\}$ of observable effect classes. Every realized effect in this profile maps to an event in $\SigmaEff$. The enforcement barrier mediates every such event before execution.
\end{quote}

In concrete deployments, the required mediation boundary depends on the chosen
effect profile. For the HTTP-level web profile in this paper, theorem-level
mediation must occur at a typed tool wrapper, browser/API shim, HTTP proxy, or
pre-TLS application boundary that emits normalized $\Net$ events. Kernel hooks
can enforce lower-level target, file, and process effects, but by themselves do
not expose HTTP methods under TLS. The theorem covers only mediated
$\SigmaEff$ effects; bypasses are out of scope.

\paragraph{Load-bearing extraction assumptions.} Applying the formal theorem to a deployment also requires three explicit obligations:
\begin{squarelist}
\item \textbf{Event-mapping completeness}: each realized in-profile effect is represented by a correct event $e=\langle\sigma,p,t,i\rangle$.
\item \textbf{Mediation coverage}: each such event reaches the monitor before execution and cannot be executed after denial.
\item \textbf{Disclosure extraction}: $\mathsf{flows}(e)$ includes every disclosure that the deployment chooses to model at the event boundary; disclosures outside this extractor, including content hidden inside allowed payloads, are outside theorem scope.
\end{squarelist}
These are deployment obligations, not mechanized lemmas.

\subsection{Dual-Layer Enforcement Architecture}

Our enforcement uses two defense-in-depth layers.

\paragraph{Agentic layer (mitigation).} Intercepts tool calls at the planner interface. Performs semantic analysis: intent classification, payload inspection, context-aware policy checks. Provides explainable denials and can request human-in-the-loop approval for ambiguous cases. This layer offers \emph{best-effort mitigation}. It can reduce attack surface, but adversarial inputs may evade semantic classifiers, and the layer does not carry formal guarantees.

\paragraph{Typed mediation layer (enforcement).} Mediates effects at the boundary
appropriate to the event profile. For the HTTP-level $\Net$ profile, this is an
application-level mediator such as a typed tool wrapper, browser/API shim, HTTP
proxy, or pre-TLS hook that can observe the normalized HTTP method and accounting
fields. For $\FS$ and $\Exec$, kernel hooks or sandbox interposition may provide
the corresponding file and process events. This is the enforcement barrier used
by the soundness theorem: it sees the modeled effect parameters and enforces the
contract regardless of how the agent arrived at the action.

\paragraph{Defense in depth.} The two layers are complementary: if the agentic layer fails (e.g.\ adversarial prompt injection bypasses semantic checks), the typed mediator still denies out-of-scope mediated effects; if the mediator has gaps, the agentic layer may still block some unsafe requests. The agentic layer adds context; typed mediation is where the theorem applies.

\paragraph{Why keep a non-proved layer?} The agentic layer remains useful for operational reasons even without theorem-level guarantees. First, it can catch some payload-semantic abuse early (prompt-injection strings, suspicious argument structures, malformed tool parameters) before those requests reach system-level mediation. Second, it can reduce noise and cost by stopping clearly invalid actions earlier in the stack. Third, it improves operator control and auditability through human-readable denials and escalation hooks. These are mitigation benefits, not proof obligations, and should be read as defense-in-depth engineering~\cite{abdelnabi2025firewalls,evertz2025patterns}.

\paragraph{Evaluation scope and methods.} This paper does \emph{not} claim empirical effectiveness numbers for the agentic layer (no precision/recall, no standalone bypass-rate claims). Our formal results are intentionally independent of this layer. In the architecture, we model representative agentic-layer methods, namely payload and argument inspection, intent/policy classification, and optional human-in-the-loop approval for high-risk actions. A full evaluation of these methods belongs to a separate empirical study with explicit attack sets and metrics, including pre-block rate, false-block rate, and end-to-end decision latency.

\paragraph{Scope of this paper.} This paper develops the \emph{formal model} for alignment contracts, not the enforcement implementation. Our contribution is a contract language with decidable checking, refinement and one-way composition semantics, and layer-agnostic specification. After event extraction and primitive predicate evaluation are available, the core admissibility check is computable in $O(|\mathsf{Res}|\cdot|\Tr|)$ offline and $O(|\mathsf{Res}|)$ online with cached counters for bounded $\mathsf{flows}(e)$. End-to-end enforcement also pays for parsing, canonicalization, DNS/IP compilation, hashing, path normalization, and flow extraction. A deployment can parse a signed contract artifact and use it to make allow/deny decisions at both the agentic layer (richer context, best-effort) and the typed mediation layer (effect-parameter checks under EOA). Architecture-graph predicates are treated abstractly via a safety predicate and trace bridge assumptions; the adaptation result is therefore presented as an extensibility schema, not as a universal adaptation guarantee.

\paragraph{Interpretation.} This framework does not solve agent alignment; it solves \emph{effect-constrained execution} for agentic systems. Alignment contracts provide the formal language for specifying ``what the agent may do,'' and enforcement systems can implement this language to deny inadmissible effects at runtime. Even when the model is adversarial, monitor-realized traces remain within the authored contract as long as actions are correctly mapped into $\SigmaEff$ and mediated. This shifts theorem-level trust from learned components (the LLM) to a small, auditable enforcement barrier.

\section{The Alignment Contract}
\label{sec:calculus}

The preceding sections established that offensive agents need asymmetric constraints: authorized capability against designated targets and denial outside that boundary. Existing mechanisms provide useful enforcement policies, but they do not center this objective as an effect-level contract semantics in our theorem setting. We now develop the contract model that does.

The model must provide effect-boundary semantics, decidable admissibility for concrete events, and modular policies through refinement and compatible composition.

\subsection{Effect Model}

\begin{definition}[Effect Alphabet\leanmark]
\label{def:alphabet}
An \emph{effect alphabet} $\SigmaEff$ is a finite set of effect classes. Each class $\sigma \in \SigmaEff$ has an associated parameter space $\mathsf{Params}(\sigma)$.
We write
\[
\mathsf{Desc} := \Sigma_{\sigma \in \SigmaEff}\mathsf{Params}(\sigma)
\]
for the effect-descriptor domain.
\end{definition}

\paragraph{Base Types.}
We assume standard base sets:
\begin{squarelist}
\item $\mathsf{IP}$ (IPv4/IPv6 addresses), $\mathsf{Domain}$ (DNS names), and $\mathsf{Host} := \mathsf{IP} \cup \mathsf{Domain}$.
\item $\mathsf{Path}$ (canonicalized filesystem paths), $\mathsf{Hash}$ (SHA-256 digests), $\mathsf{Args}$ (argument sequences), and $\mathsf{AgentID}$ (agent identifiers).
\item $\mathsf{Port} := \{0, \ldots, 65535\}$.
\item $\mathsf{L4Proto} := \{\mathsf{TCP}, \mathsf{UDP}, \ldots\}$.
\item $\mathsf{HTTPMethod} := \{\mathsf{GET}, \mathsf{POST}, \ldots\}$.
\item finite $\mathsf{Resource}$ identifiers for budgets, finite $\mathsf{DataClass}$ labels for data sensitivity, and finite $\mathsf{Sink}$ output destinations.
\end{squarelist}
For optional fields, define $T? := T \cup \{\bot\}$ with $\bot \notin T$.
For example, $\mathsf{HTTPMethod}? = \mathsf{HTTPMethod} \cup \{\bot\}$.

\noindent The concrete web-pentest profile in this paper uses
\[
\SigmaEff = \{\Net, \FS, \Exec\}.
\]
Each parameter space is a record type:
\begin{squarelist}
\item $\mathsf{Params}(\Net)$ has fields $\mathsf{dst}:\mathsf{Host}$, $\mathsf{port}:\mathsf{Port}$, and $\mathsf{l4}:\mathsf{L4Proto}$, plus optional $\mathsf{method}:\mathsf{HTTPMethod}?$ and $\mathsf{sizeBytes}:\mathbb{N}$ (destination as IP or domain, port, transport protocol, optional HTTP method, and deployment-supplied byte-accounting label).
\item $\mathsf{Params}(\FS) = \{\mathsf{path}: \mathsf{Path}, \mathsf{mode}: \{\mathsf{R}, \mathsf{W}, \mathsf{X}\}\}$ (path, access mode)
\item $\mathsf{Params}(\Exec) = \{\mathsf{hash}: \mathsf{Hash}, \mathsf{args}: \mathsf{Args}\}$ (binary hash, arguments)
\end{squarelist}
\noindent \textbf{Positioning.} The alignment-contract calculus is a framework parameterized by a finite effect alphabet and class-specific parameter spaces. The profile above is one concrete instantiation for tool-using web penetration testing. In this setting, $\FS$ and $\Exec$ are included because web agents typically invoke local tools and manipulate local artifacts in addition to making network requests.

The optional $\mathsf{method}$ field in $\Net$ allows contracts to constrain HTTP verbs when applicable; for raw TCP/UDP traffic, this field is $\bot$. The running example uses HTTP-level $\Net$ events only, so raw transport events with $\mathsf{method}=\bot$ are not admitted by Example~\ref{ex:contract}. A lower-layer profile can admit DNS, TCP, or TLS prerequisite events by adding explicit descriptor rules.

Lower-layer mediation can also be modeled, for example with Ethernet-level constraints. This requires additional effect classes and parameter schemas, plus corresponding target extraction and policy predicates. The theorems remain schema-level once those extensions satisfy the same mediation and observability assumptions used throughout this paper.

Domain patterns in authored contracts are compiled before enforcement. The compiler resolves CNAME chains according to a fixed policy, rejects private or reserved IP ranges unless explicitly allowlisted, pins the resulting IP/port predicate for a declared refresh interval, and treats DNS-rebinding changes as a new instantiation obligation. In examples we still write host-pattern predicates for readability, but these denote compiled target predicates. The target for $\Net$ is destination and port; byte counts remain descriptor/accounting data. In Lean, this class-specific schema is represented by a normalized descriptor carrier (Appendix~\ref{sec:appendix-mechanization-tables}, Table~\ref{tab:mechanization-map}).

\begin{definition}[Event\leanmark]
\label{def:event}
An \emph{event} is a tuple $e = \langle \sigma, p, t, i \rangle$ where $\sigma \in \SigmaEff$ is the effect class, $p \in \mathsf{Params}(\sigma)$ is the parameter record, $t \in \mathbb{N}$ is the mediation-time timestamp, and $i \in \mathsf{AgentID}$ is the agent identifier.
\end{definition}

\begin{definition}[Trace\leanmark]
\label{def:trace}
A \emph{trace} $\Tr$ is a finite sequence of events $e_1, \ldots, e_n$. We write $\Tr_{<k}$ for the prefix $e_1, \ldots, e_{k-1}$, $|\Tr| \in \mathbb{N}$ for length, and $\Tr' \sqsubseteq \Tr$ if $\Tr'$ is a prefix of $\Tr$ (i.e.\ $\Tr = \Tr' \cdot \Tr''$ for some finite $\Tr''$). Let $\epsilon$ denote the empty trace. We denote the event domain by $\mathsf{Event}$. (Timestamp monotonicity is a deployment-side well-formedness condition, not used by the mechanized core theorems.)
\end{definition}

\subsection{Contract Syntax}

\paragraph{Target Extraction.} For the $\{\Net,\FS,\Exec\}$ profile used in this paper, $\mathsf{Target}$ is a disjoint union (tagged by effect class):
\[
\mathsf{Target} := \mathsf{Net}_T(\mathsf{Host}, \mathsf{Port}) \uplus \mathsf{Fs}_T(\mathsf{Path}) \uplus \mathsf{Exec}_T(\mathsf{Hash})
\]
The extraction function $\mathsf{target} : \mathsf{Event} \to \mathsf{Target}$ is:
\begin{align*}
\mathsf{target}(\langle \Net, p, t, i \rangle) &:= \mathsf{Net}_T(p.\mathsf{dst}, p.\mathsf{port}) \\
\mathsf{target}(\langle \FS, p, t, i \rangle) &:= \mathsf{Fs}_T(p.\mathsf{path}) \\
\mathsf{target}(\langle \Exec, p, t, i \rangle) &:= \mathsf{Exec}_T(p.\mathsf{hash})
\end{align*}

\begin{definition}[Alignment Contract\leanmark]
\label{def:contract}
An \emph{alignment contract} is a tuple
\[
\Contract = \langle S, E_{\mathsf{allow}}, E_{\mathsf{forbid}}, B, D, \mathsf{Res}, \mathsf{cost}, \mathsf{flows}\rangle
\]
where:
\begin{squarelist}
\item $S : \mathsf{Target} \to \mathsf{Prop}$ is a decidable scope predicate over target identifiers.
\item $E_{\mathsf{allow}}, E_{\mathsf{forbid}} : \mathsf{Desc} \to \mathsf{Prop}$ are decidable effect-descriptor predicates. $E_{\mathsf{forbid}}$ has precedence in admissibility.
\item $B : \mathsf{Resource} \to \mathbb{N}$ is the budget function.
\item $D : \mathsf{DataClass} \to \mathsf{Sink} \to \mathsf{Prop}$ is the disclosure predicate.
\item $\mathsf{Res}$ is a finite list of resources tracked by budget checks (interpreted extensionally as a set for membership).
\item $\mathsf{cost} : \mathsf{Resource} \to \mathsf{Event} \to \mathbb{N}$ is the accounting function.
\item $\mathsf{flows} : \mathsf{Event} \to \mathcal{P}(\mathsf{DataClass} \times \mathsf{Sink})$ extracts declared flows per event (mechanized as a finite list representation with the same membership semantics). If disclosure depends on history or taint state, that state must be compiled into the emitted event label or handled by an extended contract interface; the core theorem only checks declared event-level flows.
\end{squarelist}
We write $(\sigma, p) \in E_{\mathsf{allow}}$ as shorthand for $E_{\mathsf{allow}}(\sigma, p)$ and similarly for $E_{\mathsf{forbid}}$. Set operators on predicates are lifted pointwise: $P \subseteq Q$ means $\forall x.\, P(x) \Rightarrow Q(x)$.
\end{definition}

\begin{example}[Web Pentest Contract]
\label{ex:contract}
The engagement from \S\ref{sec:example} is formalized as:
\begin{align*}
S(\mathsf{Net}_T(h, \mathit{pt})) &\text{ holds iff } (h,\mathit{pt}) \in \mathsf{CompiledScope} \\
E_{\mathsf{allow}} &= \{(\Net, p) \mid p.\mathsf{l4} = \mathsf{TCP} \land \\
&\qquad p.\mathsf{method} \in \{\mathsf{GET}, \mathsf{HEAD}\}\} \\
E_{\mathsf{forbid}} &= \{(\FS, \_)\} \cup \{(\Exec, \_)\} \\
B(\mathsf{requests}) &= 10000,\ \mathsf{Res} = [\mathsf{requests}] \\
D(\mathsf{VulnReport}, \mathsf{DashboardAPI}) &\text{ holds}
\end{align*}
Here $\mathsf{CompiledScope}$ is the authoring-level pattern \texttt{*.example.com} after DNS/CNAME/private-IP policy compilation. Non-$\Net$ targets are out-of-scope, and this example intentionally denies all $\FS$ events rather than merely writes. A deployment that needs read-only cache access must add explicit $\FS$ scope and allow rules. The budget above is cumulative over the engagement, matching Definition~\ref{def:admissible}. Sliding-window rate limits and engagement windows require extensions and are not part of the mechanized core. In the core example, $\mathsf{cost}_{\mathsf{requests}}(e)=1$ and disclosure classification is event-level in the mechanized core.

\paragraph{Third-party resource loading.}
If an in-scope target requires external JavaScript/CSS/media domains, this must be represented explicitly in the authored contract. A common pattern is to define a pre-approved dependency host set $\mathsf{DepHosts}$ and set
\[
S(\mathsf{Net}_T(h,\mathit{pt})) \Leftrightarrow h \in \mathsf{MainHosts} \cup \mathsf{DepHosts},
\]
while keeping stricter rules for $\mathsf{DepHosts}$ (for example, allow only GET/HEAD and forbid disclosure sinks or credential-bearing requests). If dependency domains are not pre-authorized, they are out-of-scope and denied.
\end{example}

\subsection{Admissibility and Satisfaction}

\noindent We define $\mathsf{cost}_r : \mathsf{Event} \to \mathbb{N}$ as the cost of an event for resource $r$ (e.g.\ $\mathsf{cost}_{\mathsf{requests}}(e) = 1$ for all events, $\mathsf{cost}_{\mathsf{bytes}}(e) = p.\mathsf{sizeBytes}$ for $\Net$ events). The field $\mathsf{sizeBytes}$ is a pre-admission accounting label supplied by the event extractor, such as request bytes or a conservative byte reservation; it is not payload content. Charging response bytes requires a separate post-response accounting event or reservation scheme, which is outside the core example. For disclosure, let $\mathsf{flows} : \mathsf{Event} \to \mathcal{P}(\mathsf{DataClass} \times \mathsf{Sink})$ return declared data flows associated with an event.

\begin{definition}[Admissibility\leanmark]
\label{def:admissible}
An event $e = \langle \sigma, p, t, i \rangle$ is \emph{admissible} with respect to finite trace prefix $\Tr$ and contract $\Contract$, written $\Adm(e, \Tr, \Contract)$, iff all of the following hold:
\begin{enumerate}
\item \textbf{Scope}: $S(\mathsf{target}(e))$
\item \textbf{Allowed}: $(\sigma, p) \in E_{\mathsf{allow}}$
\item \textbf{Not forbidden}: $(\sigma, p) \notin E_{\mathsf{forbid}}$
\item \textbf{Budget}: $\forall r \in \mathsf{Res}.\; \mathsf{cost}_r(e) + \sum_{j=1}^{|\Tr|} \mathsf{cost}_r(e_j) \leq B(r)$
\item \textbf{Disclosure}: $\forall (c,s) \in \mathsf{flows}(e).\; D(c,s)$
\end{enumerate}
Architecture-level assumptions are handled separately in \S\ref{sec:adaptation} via an abstract safety predicate and trace bridge assumptions.
\end{definition}

\begin{definition}[Contract Satisfaction\leanmark]
\label{def:satisfaction}
A finite trace $\Tr$ \emph{satisfies} contract $\Contract$, written $\Tr \models \Contract$, iff:
\[
\forall k \in \{1, \ldots, |\Tr|\}.\; \Adm(e_k, \Tr_{<k}, \Contract)
\]
\end{definition}

A critical question for enforcement is whether contract violations have finite witnesses. The mechanized core is finite-trace only: if a finite trace violates the contract, some finite prefix already exhibits the violation. This is weaker than a full infinite-trace safety theorem, which we do not claim.

\begin{proposition}[Finite-Trace Bad-Prefix Witness\leanmark]
\label{prop:safety}
For finite traces, contract satisfaction has finite bad-prefix witnesses: if $\Tr \not\models \Contract$, then there exists a prefix $\Tr' \sqsubseteq \Tr$ such that $\Tr' \not\models \Contract$.
\end{proposition}

\begin{proof}
If $\Tr \not\models \Contract$, then by Definition~\ref{def:satisfaction} there
exists $k \le |\Tr|$ such that $\neg \Adm(e_k,\Tr_{<k},\Contract)$. Let $\Tr'$
be the prefix containing the first $k$ events. Then $\Tr' \sqsubseteq \Tr$ by
construction and $\Tr' \not\models \Contract$ because the same violation occurs
at position $k$. This is mechanized as \path{Formal.safety_finiteWitness}.
\end{proof}

\begin{example}[Trace Evaluation]
\label{ex:trace}
Consider the trace from \S\ref{sec:example} under the contract of Example~\ref{ex:contract}. Define parameter records:
{\footnotesize
\begin{align*}
p_1 &= \{\mathsf{dst}{=}\texttt{api.example.com}, \mathsf{port}{=}443, \mathsf{l4}{=}\mathsf{TCP}, \\
&\qquad \mathsf{method}{=}\mathsf{GET}, \mathsf{sizeBytes}{=}1024\} \\
p_2 &= \{\mathsf{dst}{=}\texttt{login.example.com}, \mathsf{port}{=}443, \mathsf{l4}{=}\mathsf{TCP}, \\
&\qquad \mathsf{method}{=}\mathsf{HEAD}, \mathsf{sizeBytes}{=}512\} \\
p_3 &= \{\mathsf{dst}{=}\texttt{api.example.com}, \mathsf{port}{=}443, \mathsf{l4}{=}\mathsf{TCP}, \\
&\qquad \mathsf{method}{=}\mathsf{GET}, \mathsf{sizeBytes}{=}256\} \\
p_4 &= \{\mathsf{dst}{=}\texttt{admin.example-corp.com}, \mathsf{port}{=}443, \mathsf{l4}{=}\mathsf{TCP}, \\
&\qquad \mathsf{method}{=}\mathsf{GET}, \mathsf{sizeBytes}{=}256\}
\end{align*}}
The concrete HTTP path (e.g.\ \path{/admin}) is not a field of the core
$\Net$ descriptor; these records therefore represent the trace up to the modeled
destination, method, port, protocol, accounting size, timestamp, and agent fields.
\begin{squarelist}
\item $e_1 = \langle \Net, p_1, t_1, \mathsf{Scanner} \rangle$: $\Adm(e_1, \epsilon, \Contract)$ holds (in scope, GET allowed, budget ok)
\item $e_2 = \langle \Net, p_2, t_2, \mathsf{Scanner} \rangle$: $\Adm(e_2, [e_1], \Contract)$ holds
\item $e_3 = \langle \Net, p_3, t_3, \mathsf{Scanner} \rangle$: $\Adm(e_3, [e_1,e_2], \Contract)$ holds
\item $e_4 = \langle \Net, p_4, t_4, \mathsf{Scanner} \rangle$: $\neg \Adm(e_4, [e_1,e_2,e_3], \Contract)$ since the destination is outside the compiled scope.
\end{squarelist}
The trace $[e_1, e_2, e_3]$ satisfies $\Contract$; the trace $[e_1, e_2, e_3, e_4]$ does not. The prefix $[e_1, e_2, e_3, e_4]$ is a finite witness to the violation.
\end{example}

\paragraph{Notation Summary.} Appendix Table~\ref{tab:notation} summarizes the key notation introduced in this section. All theorem claims built on this formal core are mapped to Lean theorem names in Appendix~\ref{sec:appendix-mechanization-tables}.

\section{Enforcement Soundness}
\label{sec:soundness}

The framework gives us a language for specifying contracts and a definition of satisfaction. But specification is not enforcement. The central question is whether a monitor-realized trace satisfies its contract regardless of the proposed agent actions.

This is the soundness question, and its answer is the central formal result. Under EOA, the proof quantifies over proposed traces and shows that the monitor-realized trace satisfies the contract. The result makes no behavioral assumption about the LLM; its assumptions are about mediation and event extraction.

\subsection{Reference Monitor Semantics}

\begin{definition}[Reference Monitor\leanmark]
\label{def:monitor}
A \emph{reference monitor} $\Mon_\Contract$ for contract $\Contract$ is a deterministic function over intercepted events. Each event has form $e = \langle \sigma, p, t, i \rangle \in \mathsf{Event}$, where $(\sigma,p,i)$ comes from event extraction and the timestamp $t$ is assigned by the environment at mediation time. Given event $e$ and current trace $\Tr$:
\[
\Mon_\Contract(e, \Tr) =
\begin{cases}
\mathsf{Allow}(e) & \text{if } \Adm(e, \Tr, \Contract) \\
\mathsf{Deny}(e, \mathsf{reason}) & \text{otherwise}
\end{cases}
\]
The monitor operates over intercepted effects, not planner intentions. A single planner or tool action, such as fetching a URL, may generate a finite sequence of mediated events after parsing and normalization. The formal monitor consumes that sequence incrementally; each intercepted effect is one event.
\end{definition}

\noindent\textit{Remark (Observability).} Admissibility $\Adm(e, \Tr, \Contract)$ is computable from observable data only: the proposed event $e$, the trace $\Tr$ of past events, and the contract $\Contract$. No internal agent state is required; enforcement is purely behavioral.

\begin{definition}[Monitored System]
\label{def:system}
Let $A$ be an arbitrary agent whose actions induce a finite proposed event trace after extraction. The \emph{monitored system} $\Sys = \Mon_\Contract(A)$ is the composition where every extracted event passes through $\Mon_\Contract$. Allowed events are appended; denials append nothing.
\end{definition}

\subsection{Supporting Lemmas and Soundness}

For any finite proposed trace $P$, let $\mathsf{Run}_\Contract(P)$ be the realized
subtrace kept by iterating the monitor from empty history.

\begin{lemma}[Prefix Transitivity\leanmark]
\label{lem:prefix-preserve}
If $\Tr_1 \sqsubseteq \Tr_2$ and $\Tr_2 \sqsubseteq \Tr_3$, then $\Tr_1 \sqsubseteq \Tr_3$.
\end{lemma}

\begin{proof}
If $\Tr_2=\Tr_1\cdot s_1$ and $\Tr_3=\Tr_2\cdot s_2$ (prefix definition), then
$\Tr_3=\Tr_1\cdot (s_1\cdot s_2)$, so $\Tr_1 \sqsubseteq \Tr_3$.
(Mechanized as \path{Formal.prefix_trans}.)
\end{proof}

\begin{lemma}[Monitor Run Determinism\leanmark]
\label{lem:monitor-det}
For any contract $\Contract$ and proposed traces $P_1,P_2$, if $P_1=P_2$ then
$\mathsf{Run}_\Contract(P_1)=\mathsf{Run}_\Contract(P_2)$.
\end{lemma}

\begin{proof}
If $P_1=P_2$, replacing $P_1$ by $P_2$ in the deterministic function
$\mathsf{Run}_\Contract$ gives equal outputs. Mechanized as
\path{Formal.monitor_state_deterministic}.
\end{proof}

\begin{theorem}[Enforcement Soundness\leanmark]
\label{thm:soundness}
For every finite $P$:
\[
\mathsf{Run}_\Contract(P) \models \Contract
\]
\end{theorem}

\begin{proof}
By induction on finite $P$. Base case $P=\epsilon$: $\mathsf{Run}_\Contract(P)=\epsilon$
and $\epsilon \models \Contract$. Step case: let $\Tr_0=\mathsf{Run}_\Contract(P_0)$
with $\Tr_0 \models \Contract$ by induction. For the next proposal, deny leaves $\Tr_0$
unchanged; allow gives $\Adm(e,\Tr_0,\Contract)$ and appends $e$, so $\Tr_0\cdot[e]
\models \Contract$ by Definition~\ref{def:satisfaction}. (Mechanized as
\path{Formal.enforcement_soundness}.)
\end{proof}

\begin{corollary}[Agent-Behavior Independence]
\label{cor:agnostic}
If deployed mediation implements $\mathsf{Run}_\Contract$ faithfully (EOA), then realized effects satisfy $\Contract$ independent of the agent's internal policy, including adversarial or prompt-injected models.
\end{corollary}

\paragraph{Trusted Computing Base.} The TCB consists of: (1) the reference monitor implementation, (2) the mediation points that ensure all $\SigmaEff$-effects are intercepted (the EOA), and (3) the contract specification $\Contract$. The LLM and agent logic are explicitly \emph{outside} the TCB.

\section{Contract Algebra}
\label{sec:algebra}

Soundness covers one contract. Deployments often combine overlapping constraints: engagement scope, organizational policy, regulation, and per-agent limits. A single monolithic contract is possible, but modular contracts improve ownership boundaries, change isolation, reuse, and auditability. We therefore study refinement and a compatible composition operator, while stating only the one-way soundness theorem that is mechanized.
Section~\ref{sec:soundness} established semantic soundness; decidability supports modular runtime reasoning.

\paragraph{Scope.} This section concerns \emph{trace satisfaction} (Definition~\ref{def:satisfaction}) in the mechanized core. Architecture-level obligations are handled in \S\ref{sec:adaptation} through the abstract $\mathsf{Safe}$ predicate and bridge assumptions.

\subsection{Refinement}

\begin{definition}[Contract Refinement\leanmark]
\label{def:refinement}
Contract $\Contract'$ \emph{refines} contract $\Contract$, written $\Contract' \sqsubseteq \Contract$, iff:
\begin{enumerate}
\item $S' \subseteq S$ (scope is narrower or equal)
\item $E'_{\mathsf{allow}} \subseteq E_{\mathsf{allow}}$ (fewer allowed effects)
\item $E_{\mathsf{forbid}} \subseteq E'_{\mathsf{forbid}}$ (more forbidden effects)
\item $\forall r.\; B'(r) \leq B(r)$ (tighter budgets)
\item $D' \subseteq D$ (stricter disclosure)
\item $\forall r.\; r \in \mathsf{Res} \Rightarrow r \in \mathsf{Res}'$ (every resource checked by $\Contract$ is checked by $\Contract'$)
\item $\mathsf{cost}' = \mathsf{cost}$ and $\mathsf{flows}' = \mathsf{flows}$ (same accounting/flow extraction semantics)
\end{enumerate}
\end{definition}

\begin{theorem}[Refinement Soundness\leanmark]
\label{thm:refinement}
If $\Contract' \sqsubseteq \Contract$ and $\Tr \models \Contract'$, then $\Tr \models \Contract$.
\end{theorem}

\begin{proof}
Assume $\Tr \models \Contract'$ and $\Contract' \sqsubseteq \Contract$. To prove
$\Tr \models \Contract$, fix $k \in \{1,\ldots,|\Tr|\}$ and let
$e_k=\langle \sigma_k,p_k,t_k,i_k \rangle$. From $\Tr \models \Contract'$, we have
$\Adm(e_k,\Tr_{<k},\Contract')$.

Scope and allow are immediate from refinement clauses (1) and (2).
For forbid, clause (3) gives $E_{\mathsf{forbid}} \subseteq E'_{\mathsf{forbid}}$,
so $(\sigma_k,p_k)\notin E'_{\mathsf{forbid}}$ implies
$(\sigma_k,p_k)\notin E_{\mathsf{forbid}}$.

For budget, fix $r \in \mathsf{Res}$. By clause (6), $r \in \mathsf{Res}'$.
Admissibility under $\Contract'$ gives
\[
\mathsf{cost}'_r(e_k)+\sum_{j=1}^{k-1}\mathsf{cost}'_r(e_j)\le B'(r).
\]
By clause (7), $\mathsf{cost}'=\mathsf{cost}$, so the left side is the cumulative
cost under $\Contract$; by clause (4), $B'(r)\le B(r)$. Hence
\[
\mathsf{cost}_r(e_k)+\sum_{j=1}^{k-1}\mathsf{cost}_r(e_j)\le B(r).
\]

For disclosure, let $(c,s)\in \mathsf{flows}(e_k)$. By clause (7), this is also in
$\mathsf{flows}'(e_k)$; admissibility under $\Contract'$ gives $D'(c,s)$, and clause
(5) ($D' \subseteq D$) gives $D(c,s)$.

Thus $\Adm(e_k,\Tr_{<k},\Contract)$ holds. Since $k$ was arbitrary,
$\Tr \models \Contract$. This proof is mechanized as
\path{Formal.refinement_soundness}, using the helper theorem
\path{Formal.admissible_of_refines}.
\end{proof}

\subsection{Composition}

\begin{definition}[Contract Composition\leanmark]
\label{def:composition}
The \emph{compatible composition} of contracts $\Contract_1$ and $\Contract_2$, written $\mathsf{compose}(\Contract_1,\Contract_2)$, is:
\begin{align*}
S_{12} &= S_1 \cap S_2 \\
E_{\mathsf{allow},12} &= E_{\mathsf{allow}, 1} \cap E_{\mathsf{allow}, 2} \\
E_{\mathsf{forbid},12} &= E_{\mathsf{forbid}, 1} \cup E_{\mathsf{forbid}, 2} \\
B_{12}(r) &= \min(B_1(r), B_2(r)) \\
D_{12} &= D_1 \cap D_2 \\
\mathsf{Res}_{12} &= \mathsf{Res}_1,\quad
\mathsf{cost}_{12} = \mathsf{cost}_1,\quad
\mathsf{flows}_{12} = \mathsf{flows}_1
\end{align*}
\end{definition}

\noindent This definition intentionally projects the resource, cost, and flow
components from $\Contract_1$. The projection is sound only under the
compatibility condition below; without compatibility, the composition operator should not be
read as semantic conjunction.

\begin{definition}[Compatibility\leanmark]
\label{def:compose-compat}
$\mathsf{Compat}(\Contract_1,\Contract_2)$ holds iff
$\mathsf{Res}_1=\mathsf{Res}_2$, $\mathsf{cost}_1=\mathsf{cost}_2$, and
$\mathsf{flows}_1=\mathsf{flows}_2$.
\end{definition}

\begin{theorem}[Composition Soundness\leanmark]
\label{thm:composition}
If $\mathsf{Compat}(\Contract_1,\Contract_2)$ and
$\Tr \models (\mathsf{compose}(\Contract_1,\Contract_2))$, then
$\Tr \models \Contract_1$ and $\Tr \models \Contract_2$.
\end{theorem}

\begin{proof}
Assume $\mathsf{Compat}(\Contract_1,\Contract_2)$ and
$\Tr \models (\mathsf{compose}(\Contract_1,\Contract_2))$. Fix
$k \in \{1,\ldots,|\Tr|\}$ and let $e_k$ be the $k$-th event of $\Tr$.
By Definition~\ref{def:satisfaction}, we have
$\Adm(e_k,\Tr_{<k},\mathsf{compose}(\Contract_1,\Contract_2))$.
For $\Contract_1$, admissibility projects directly from this fact:
scope/allow from intersections, and not-forbid from union.
For budget, use the fact that the composed budget is at most $B_1(r)$ together
with left-projected resources/accounting. For disclosure, membership in the
intersection policy implies membership in $D_1$, again using left-projected
flow extraction.
For $\Contract_2$, the first three clauses are symmetric. Compatibility ensures
the same resource set and the same cost/flow extractors for $\Contract_1$ and
$\Contract_2$, so the budget bound and disclosure membership transport from
$\mathsf{compose}(\Contract_1,\Contract_2)$ to $\Contract_2$.
Hence both admissibility obligations hold for $e_k$.
Since $k$ was arbitrary:
\[
\Tr \models \Contract_1
\quad\text{and}\quad
\Tr \models \Contract_2.
\]
\emph{Mechanization note.} This proof maps to the one-way composition theorem
listed in Table~\ref{tab:theorem-coverage}. We do not claim the converse:
satisfying both components is not stated to imply satisfying their composition.
\end{proof}

\subsection{Decidability}

Algebraic properties are not useful for runtime enforcement if checking them requires solving undecidable problems. We therefore state decidability assumptions and the core cost bound.

\begin{proposition}[Decidability of Admissibility\leanmark]
\label{prop:decidable}
For contract $\Contract=\langle S,E_{\mathsf{allow}},E_{\mathsf{forbid}},B,D,\mathsf{Res},\mathsf{cost},\mathsf{flows}\rangle$ over finite $\SigmaEff$, with decidable predicates $S$, $E_{\mathsf{allow}}$, $E_{\mathsf{forbid}}$, and $D$, finite tracked resource list $\mathsf{Res}$, and computable functions $B$, $\mathsf{cost}_r$, and $\mathsf{flows}$, the admissibility predicate $\Adm(e, \Tr, \Contract)$ is decidable.
After event extraction and canonicalization, and assuming unit-cost evaluation of these primitive predicates/functions on concrete inputs, one core admissibility check has the following cost:
\[
m_e := |\mathsf{flows}(e)|.
\]
\begin{squarelist}
\item \textbf{Offline}: Given only $(e, \Tr, \Contract)$, decidable in $O(|\mathsf{Res}| \cdot |\Tr| + m_e)$.
\item \textbf{Online}: Cached prefix sums give $O(|\mathsf{Res}| + m_e)$ per event.
\end{squarelist}
If $m_e$ is bounded by a small constant, the core check is $O(|\mathsf{Res}| \cdot |\Tr|)$ offline and $O(|\mathsf{Res}|)$ online.
\end{proposition}

\begin{proof}
We analyze each condition in Definition~\ref{def:admissible}:

\textbf{(1) Scope:} Evaluate $S(\mathsf{target}(e))$. Decidable by assumption on $S$.

\textbf{(2)--(3) Effect predicates:} Evaluate $E_{\mathsf{allow}}(\sigma, p)$ and $E_{\mathsf{forbid}}(\sigma, p)$. Decidable by assumption.

\textbf{(4) Budget:} For each $r \in \mathsf{Res}$, compute and compare
$\mathsf{cost}_r(e) + \sum_{e' \in \Tr} \mathsf{cost}_r(e') \leq B(r)$.
Since $\mathsf{Res}$ is finite and $B,\mathsf{cost}_r$ are computable, this check is decidable.
Offline this is $O(|\mathsf{Res}| \cdot |\Tr|)$; online with cached prefix sums it is $O(|\mathsf{Res}|)$.

\textbf{(5) Disclosure:} Check $\forall (c,s) \in \mathsf{flows}(e).\; D(c,s)$.
With finite $\mathsf{flows}(e)$ and decidable $D$, this is decidable and runs in
$O(m_e)$.

\textbf{Total:} offline $O(|\mathsf{Res}| \cdot |\Tr| + m_e)$; online
$O(|\mathsf{Res}| + m_e)$.

\textit{Mechanization note.} Proposition~\ref{prop:decidable} maps to the Lean
decidability instance listed in Table~\ref{tab:theorem-coverage}. Complexity
bounds are analytical.
\end{proof}

\noindent\textit{Remark (Practical complexity).} In typical deployments, predicates use pattern matching ($O(1)$), $|\mathsf{Res}|$ is small and fixed, and $m_e$ is bounded. Under these conditions, the core admissibility step is constant-time per event. End-to-end monitor latency still includes extraction, DNS/IP compilation, hashing, path normalization, flow extraction, and predicate matching outside the core theorem.

\section{Adaptation Assumption Schema}
\label{sec:adaptation}

The previous sections reason about a fixed architecture state. Real deployments change structure over time: agents may spawn sub-agents, systems may swap models or tools, and failed components may be rebound during execution. This section provides assumption plumbing, not a standalone adaptation guarantee. We define \emph{contract-preserving} adaptation rules, prove that finite sequences of such rules preserve $\mathsf{Safe}$, and transfer that fact to trace satisfaction without changing Sections~4--6's event/trace semantics.

\subsection{Architecture Graphs}

\begin{definition}[Architecture State\leanmark]
\label{def:graph}
Let $\mathsf{Graph}$ be an abstract type of architecture states.
\end{definition}

\begin{definition}[Monitor/Trace Interface\leanmark]
\label{def:monitor-instance}
Let $\mathsf{Traces} : \mathsf{Graph} \to (\mathsf{Trace} \to \mathsf{Prop})$ map each architecture state to the set of traces it can realize, and let $\mathsf{Safe} : \mathsf{Graph} \to \mathsf{Prop}$ be a safety predicate capturing ``contract-admissible architecture'' assumptions.
We write $\mathsf{Traces}(G,\Tr)$ as shorthand for $(\mathsf{Traces}(G))(\Tr)$.
\end{definition}

\begin{definition}[Safe State\leanmark]
\label{def:admissible-graph}
$G$ is contract-admissible iff $\mathsf{Safe}(G)$.
\end{definition}

\subsection{Adaptation Rules}

\begin{definition}[Adaptation Rule\leanmark]
\label{def:rewrite}
A rule is a total function $p : \mathsf{Graph} \to \mathsf{Graph}$.
\end{definition}

\begin{definition}[Preserving Rule\leanmark]
\label{def:preserving}
Rule $p$ is contract-preserving iff $\forall G.\; \mathsf{Safe}(G) \Rightarrow \mathsf{Safe}(p(G))$.
\end{definition}

\begin{proposition}[Preservation under Rule Sequences\leanmark]
\label{prop:preserving}
If $G_0$ is safe and every adaptation rule is contract-preserving for
$\mathsf{Safe}$, then the final state after applying any finite rule sequence is
safe.
\end{proposition}

\begin{proof}[Proof of Proposition~\ref{prop:preserving}]
Let the sequence be $p_1,\ldots,p_n$, and define $G_{j+1}=p_{j+1}(G_j)$ with
$G_0$ given. We prove $\mathsf{Safe}(G_n)$ by induction on $n$.

Base case $n=0$: $G_n=G_0$, so $\mathsf{Safe}(G_n)$ by assumption.

Inductive step: assume the claim for length $n$. For length $n{+}1$, by the
induction hypothesis $\mathsf{Safe}(G_n)$ holds. Since $p_{n+1}$ is
contract-preserving, $\mathsf{Safe}(G_{n+1})$ follows from
$\mathsf{Safe}(G_n)\Rightarrow \mathsf{Safe}(p_{n+1}(G_n))$.

Hence the final state after any finite preserving sequence is safe.
Mechanized as \path{Formal.preservation_under_rules}.
\end{proof}

\subsection{Extensibility Schema}

\paragraph{Bridge obligations.}
The adaptation argument uses two deployment obligations. First, $\mathsf{Safe}(G)$ must imply the mediation assumptions required by EOA. Second, $\mathsf{Safe}(G)$ must imply that every trace realizable from $G$ satisfies $\Contract$. The Lean artifact includes \path{Formal.eoa_from_safe} and \path{Formal.safe_trace_bridge} to make these assumption transports explicit; they are not standalone guarantees.
\label{lem:eoa}\label{lem:bridge}

\begin{theorem}[Adaptation Assumption Schema\leanmark]
\label{thm:compositionality}
Let $G_0$ be an initial state, let $p_1,\ldots,p_n$ be adaptation rules, and let $G_n$ be the state obtained by applying them in sequence. Assume:
\begin{enumerate}
\item $\mathsf{Safe}(G_0)$
\item every adaptation rule is contract-preserving for $\mathsf{Safe}$
\item bridge assumption: $\forall G.\, \mathsf{Safe}(G) \Rightarrow (\forall \Tr,\ \mathsf{Traces}(G,\Tr)\Rightarrow \Tr \models \Contract)$
\end{enumerate}
Then $\mathsf{Safe}(G_n)$ and every trace in $\mathsf{Traces}(G_n)$ satisfies $\Contract$.
\end{theorem}

\begin{proof}
By Proposition~\ref{prop:preserving} and assumptions (1) and (2),
$\mathsf{Safe}(G_n)$ holds.

Apply assumption (3) to $G_n$. Since $\mathsf{Safe}(G_n)$ holds, we obtain
\[
\forall \Tr,\ \mathsf{Traces}(G_n,\Tr)\Rightarrow \Tr \models \Contract.
\]
Hence $G_n$ is safe and every trace realizable from $G_n$ satisfies
$\Contract$.
Mechanized as \path{Formal.adaptation_compositionality}.
\end{proof}

\paragraph{Instantiation Note.}
Concrete architecture-graph constraints are encoded through $\mathsf{Safe}$ and
the bridge assumptions above (for example monitor placement, sandboxing,
routing, and synchronization). These obligations are explicit and must be
discharged by a concrete deployment model; they are not hidden by the
mechanized theorem. Accordingly, this section should be read as an explicit
extensibility schema parameterized by deployment assumptions.

\subsection{Concrete Deployment Instantiation}

To reduce assumption distance, the artifact also includes a concrete deployment model (\texttt{Formal.ConcreteAdaptation}) where states are finite sets of proposed traces plus a mediation flag. In this model:
\begin{enumerate}
\item $\mathsf{Safe}(d)$ is exactly ``mediation flag is true,''
\item realized traces are monitor outputs when mediation is enabled, and
\item rule preservation is checked over explicit rule lists.
\end{enumerate}

\begin{theorem}[Concrete Adaptation Soundness\leanmark]
\label{thm:concrete-adaptation}
Let $d_0$ be a concrete deployment state with mediation enabled, and let
$p_1,\ldots,p_n$ be adaptation rules such that each rule in the sequence preserves
mediation. If $d_n$ is the result of applying the sequence, then $d_n$ remains
safe and every trace realized by $d_n$ satisfies $\Contract$.
\end{theorem}

\begin{proof}
In the concrete model, ``safe'' is exactly mediation enabled. By list-preserving
adaptation, applying $p_1,\ldots,p_n$ to $d_0$ yields a final state $d_n$ that
still has mediation enabled.

Now fix any trace $\Tr$ realized by $d_n$. Since $d_n$ is safe, the concrete
safe-to-trace bridge gives $\Tr \models \Contract$.

Therefore $d_n$ is safe and all realized traces satisfy $\Contract$.
Mechanized as \path{Formal.dep_adaptation_soundness}.
\end{proof}

\section{Impossibility Results}
\label{sec:impossibility}

The preceding sections established what alignment contracts prove in the model: enforcement soundness under EOA, refinement, one-way composition, and an adaptation assumption schema. The same model also defines limits. A deployment that treats these theorems as universal agent safety would miss attacks outside the observation and mediation assumptions.

We establish three negative results. First, pre-admission forbidden-effect absence is undecidable for Turing-complete tool languages and contracts with at least one forbidden descriptor, under the supplied-reduction condition below. Second, tools that bypass mediation void the soundness theorem's premise. Third, behaviors outside the monitor observation function lie outside what effect-based monitors can enforce. These results delineate the \emph{enforcement boundary} and align with classic undecidability boundaries for protection systems~\cite{harrison1976protection}.

\subsection{Dynamic Tool Acquisition}

Runtime monitoring at the chosen mediation boundary can mediate effects from arbitrary tools; the monitor need not know in advance what effects a tool will attempt. For HTTP-level contracts this requires the typed mediation boundary described in \S\ref{sec:threat}; lower-level syscall mediation enforces only the lower-level event profile. Two distinct problems arise with dynamic tool acquisition.

\begin{definition}[Forbidden-Effect Absence Certification Problem]
\label{def:preamd}
Fix a contract $\Contract$. Let $\mathsf{events}(T,x)$ denote the possibly
infinite stream of proposed events produced by tool $T$ on input $x$. Write
$e \in \mathsf{events}(T,x)$ when $e$ occurs at some finite position in that
stream; for terminating tools this is just membership in the finite emitted
trace. Define:
for each event $e$, write
$e=\langle \sigma_e,p_e,t_e,i_e\rangle$, and let
\[
\mathsf{forbidEvent}_{\Contract}(e) \iff
E_{\mathsf{forbid}}(\sigma_e,p_e)
\]
and
\[
\begin{aligned}
\mathsf{ForbiddenAttempt}_{\Contract}(T) \iff{}&
\exists x.\;\exists e \in \mathsf{events}(T,x).\\
&\mathsf{forbidEvent}_{\Contract}(e)
\end{aligned}
\]
and
\[
\mathsf{FORBIDFREE}_{\Contract}(T) \iff
\neg \mathsf{ForbiddenAttempt}_{\Contract}(T).
\]
The forbidden-effect absence problem asks whether static admission can decide $\mathsf{FORBIDFREE}_{\Contract}(T)$. This is a subproblem, not full contract certification: a tool that is $\mathsf{FORBIDFREE}$ may still violate scope, allow, budget, or disclosure checks.
\end{definition}

The following theorem is a transfer lemma for this forbidden-effect subproblem: the reduction is the load-bearing obligation. A concrete reduction can be obtained whenever the tool language can simulate a Turing machine and emit at least one forbidden descriptor $d^\star$ if the simulated machine halts. For a tool $T_x$ that ignores its runtime input, simulates machine $M$ on the encoded instance $x$, and emits $d^\star$ exactly when $M(x)$ halts, we get
\[
\mathsf{Halts}(x) \iff \mathsf{ForbiddenAttempt}_{\Contract}(T_x)
\]
and therefore
\[
\mathsf{FORBIDFREE}_{\Contract}(T_x) \iff \neg \mathsf{Halts}(x).
\]
If forbidden-effect absence were decidable, non-halting would be decidable, and therefore halting would be decidable. Full static contract certification would be at least as hard on contract instances where all other admissibility checks are permissive. The mechanized statement below abstracts the lower bound as a supplied reduction from an undecidable predicate to $\mathsf{FORBIDFREE}$.

\begin{theorem}[Forbidden-Effect Absence under Supplied Reduction\leanmark]
\label{thm:impossibility}
Let $U$ be an undecidable predicate over instances $x$, and let
$\mathsf{encode}$ map $x$ to tools such that
\[
U(x) \iff \mathsf{FORBIDFREE}_{\Contract}(\mathsf{encode}(x)).
\]
Then the unary predicate $T \mapsto \mathsf{FORBIDFREE}_{\Contract}(T)$ is undecidable.
\end{theorem}

\begin{proof}
Suppose, for contradiction, that $\mathsf{FORBIDFREE}_{\Contract}$ is
decidable. Then there exists a decider $D_{\mathsf{ff}}$ such that for any
tool $T$,
\[
D_{\mathsf{ff}}(T)=1 \iff \mathsf{FORBIDFREE}_{\Contract}(T).
\]
Define
\[
D_U(x):=D_{\mathsf{ff}}(\mathsf{encode}(x)).
\]
By the reduction premise,
\[
U(x) \iff \mathsf{FORBIDFREE}_{\Contract}(\mathsf{encode}(x))
\iff D_U(x)=1.
\]
So $D_U$ decides $U$, contradicting undecidability
of $U$. Therefore $T \mapsto \mathsf{FORBIDFREE}_{\Contract}(T)$ is undecidable.
Mechanized as \path{Formal.preadmission_undecidable_of_reduction}.
\end{proof}

\noindent\textit{Remark.} Sound static analysis can over-approximate forbidden effect attempts, but no sound \emph{and} complete analysis exists for this subproblem on Turing-complete languages.

\begin{observation}[Mediation-Bypassing Tools]
\label{obs:bypass}
If dynamic tool acquisition permits tools that bypass the mediation layer
(kernel modules, raw hardware access, out-of-band channels), then the EOA is
violated and Theorem~\ref{thm:soundness} is inapplicable.
\end{observation}

\begin{proof}
Theorem~\ref{thm:soundness} assumes complete mediation of $\SigmaEff$-effects
(EOA). Suppose there exists a realized effect $e^\star$ from a dynamically
acquired tool that bypasses $\Mon_\Contract$. Then $e^\star$ appears in the
real execution trace but not in any trace in $\mathsf{Traces}(\Mon_\Contract(A))$.
Hence the EOA premise is false, so Theorem~\ref{thm:soundness} cannot be
applied.
\end{proof}

\begin{corollary}[Deployment Condition]
To preserve the soundness theorem's assumptions, dynamically acquired tools must be confined
within sandboxes that maintain complete mediation (for example, WebAssembly
with capability-gated syscall proxying, or seccomp-bpf filtering to the
monitored interface).
\end{corollary}

\subsection{Observation Boundary Schema}

We record an observation-boundary schema for a suppression monitor $\Mon$ and
observation function $\mathsf{obs}$. The result is \emph{parametric}: it is not
a new characterization for an arbitrary fixed monitor unless the forward and
backward obligations below are proved independently. In particular, the
backward direction usually requires constructing a canonical suppression
monitor for the target behavior or proving the obligation for the chosen
monitor.

\paragraph{Model Assumptions.} We assume:
\begin{enumerate}
\item \textbf{Suppression-only monitor}: $\Mon$ can allow or block each proposed
event; it cannot insert, reorder, or delay events.
\item \textbf{Deterministic observation-based decisions}: the decision on event
$e$ after history $\Tr$ depends only on $\mathsf{obs}(\Tr)$ and
$\mathsf{obs}(e)$.
\end{enumerate}
For a finite proposed trace $P$, let $\mathsf{Run}_{\Mon}(P)$ denote the
realized trace obtained by iterating $\Mon$ from empty history, exactly as in
Section~\ref{sec:soundness} (with $\Mon$ in place of $\Mon_\Contract$).

\begin{definition}[Observation Function\leanmark]
\label{def:obs}
An \emph{observation function} $\mathsf{obs} : \mathsf{Event} \to \mathsf{Obs}$
extracts the attributes visible to the monitor. For alignment contracts,
$\mathsf{obs}$ contains every field used by $\Adm$:
\begin{squarelist}
\item $\mathsf{obs}(\langle \Net,p,t,i\rangle)$ includes $p.\mathsf{sizeBytes}$, other NET predicate fields, and $(t,i)$.
\item $\mathsf{obs}(\langle \FS,p,t,i\rangle)$ includes FS predicate fields and
$(t,i)$.
\item $\mathsf{obs}(\langle \Exec,p,t,i\rangle)$ includes EXEC predicate fields
and $(t,i)$.
\end{squarelist}
This includes timestamps as event fields, accounting byte counts for bandwidth budgets,
and command arguments for filtering. The core contract does not include a
temporal predicate; engagement windows require an extension or an admission
wrapper. For
$\Tr=\langle e_1,\ldots,e_n\rangle$, the observed trace is
$\mathsf{obs}(\Tr)=\langle\mathsf{obs}(e_1),\ldots,\mathsf{obs}(e_n)\rangle$.

What $\mathsf{obs}$ \emph{omits} is payload content: transmitted byte strings,
file contents written, or RPC parameters beyond the method signature. These
payload fields are absent from the event model (Definition~\ref{def:event}) and
therefore absent from $\mathsf{obs}$.
\end{definition}

\paragraph{Observation adequacy for contracts.}
For the relation to contract satisfaction below, $\mathsf{obs}$ must expose all fields used by admissibility. We assume that if $\mathsf{obs}(e)=\mathsf{obs}(e')$, then $\mathsf{target}(e)=\mathsf{target}(e')$, the effect descriptors match, $\forall r.\,\mathsf{cost}_r(e)=\mathsf{cost}_r(e')$, and $\mathsf{flows}(e)=\mathsf{flows}(e')$. Equivalently, an implementation may include target, descriptor, cost labels, and flow labels directly in $\mathsf{obs}$.

\begin{definition}[Observable Behavior\leanmark]
\label{def:observable}
A behavior $\mathcal{B} \subseteq \mathsf{Event}^*$ is \emph{$\mathsf{obs}$-observable} iff:
\[
\forall \Tr_1, \Tr_2.\; \mathsf{obs}(\Tr_1) = \mathsf{obs}(\Tr_2) \implies (\Tr_1 \in \mathcal{B} \iff \Tr_2 \in \mathcal{B})
\]
\end{definition}

\begin{definition}[Enforceability\leanmark]
\label{def:enforceable}
A behavior $\mathcal{B}$ is \emph{enforceable} by monitor $\Mon$ iff:
\begin{enumerate}
\item \textbf{Soundness}: $\forall P.\; \mathsf{Run}_{\Mon}(P) \in \mathcal{B}$.
\item \textbf{Transparency}: $\forall P.\; P \in \mathcal{B} \Rightarrow \mathsf{Run}_{\Mon}(P)=P$.
\end{enumerate}
Following Ligatti et al.~\cite{ligatti2005edit}: soundness prevents violations,
and transparency preserves already-good traces.
\end{definition}

\begin{theorem}[Enforcement Boundary Schema\leanmark]
\label{thm:boundary}
Assume:
\begin{enumerate}
\item \textbf{Forward obligation}: if $\mathcal{B}$ is enforceable by $\Mon$,
then $\mathcal{B}$ is $\mathsf{obs}$-observable and has the finite bad-prefix property.
\item \textbf{Backward obligation}: if $\mathcal{B}$ is $\mathsf{obs}$-observable
and has the finite bad-prefix property, then $\mathcal{B}$ is enforceable by $\Mon$.
\end{enumerate}
Under these obligations, $\mathcal{B}$ is enforceable by $\Mon$ iff it is
$\mathsf{obs}$-observable and has the finite bad-prefix property.
\end{theorem}

\begin{proof}
$(\Rightarrow)$ If $\mathcal{B}$ is enforceable, obligation (1) yields that
$\mathcal{B}$ is $\mathsf{obs}$-observable and has the finite bad-prefix property.

$(\Leftarrow)$ If $\mathcal{B}$ is $\mathsf{obs}$-observable and a safety
behavior in this finite-trace sense, obligation (2) yields enforceability.

Hence enforceability is equivalent to
($\mathsf{obs}$-observability $\land$ finite bad-prefix property), provided the two obligations hold.
Mechanization mapping appears in Appendix Table~\ref{tab:theorem-coverage}.
\end{proof}

\paragraph{What Lies Outside.} Theorem~\ref{thm:boundary} identifies the proof obligations needed for an enforceability characterization. Behaviors that are either not $\mathsf{obs}$-observable or not finite bad-prefix properties remain outside this boundary. These are model limits.

\begin{description}
\item[Payload-based covert channels.] Information encoded in payload content (the byte string sent/written) is not in our event model and hence not in $\mathsf{obs}$. Two concrete executions may map to the same event and observation trace while carrying different payload bytes, one leaking secrets. Preventing payload-based exfiltration requires content inspection beyond effect mediation.

\item[Timing channels.] Although $t$ is in $\mathsf{obs}$, the core contract does not include temporal predicates or inter-event timing patterns. An adversary can encode information in delays between events. Preventing timing channels requires either timing normalization or explicit timing-pattern predicates in an extended contract language.

\item[Semantic intent.] Whether an action is ``malicious'' depends on the model's reasoning, not observable effects. Two traces with identical $\mathsf{obs}$-projections, one from a benign model and one adversarial, are indistinguishable. Intent is not $\mathsf{obs}$-observable.

\item[Liveness properties.] ``The agent will eventually find a vulnerability'' has no finite bad prefix; at any point, success might still come. Liveness is not a finite bad-prefix property, hence not enforceable by suppression monitors in this schema.
\end{description}

\noindent\textit{Relation to this model.} Define
\[
\mathcal{B}_\Contract := \{\Tr \mid \Tr \models \Contract\}.
\]
$\mathcal{B}_\Contract$ is $\mathsf{obs}$-observable under the observation-adequacy assumption because $\Adm$ depends only
on $\mathsf{obs}$-visible fields, and it has finite bad-prefix witnesses by
Proposition~\ref{prop:safety}. Its finite bad-prefix set is:
\[
\mathsf{Bad}_\Contract
= \{\Tr \mid \Tr \not\models \Contract\}
= \{\Tr \mid \exists k \in \{1,\ldots,|\Tr|\}.\; \neg\Adm(e_k,\Tr_{<k},\Contract)\}.
\]
Since $\Adm$ is decidable, $\mathsf{Bad}_\Contract$ is decidable, so
$\mathsf{Run}_{\Mon_\Contract}$ is computable. Thus contract enforcement lies
inside the boundary of Theorem~\ref{thm:soundness}.

\section{Related Work}
\label{sec:related}

Prior work establishes offensive capability for security-focused agents: PentestGPT, PentestAgent, MAPTA, and A1 study autonomous vulnerability discovery and exploitation across web and smart-contract settings~\cite{deng2024pentestgpt,shen2025pentestagent,david2025mapta,gervais2025a1}. Benchmarks such as AgentBench, AgentDojo, WebArena, SWE-bench, and CyberSecEval measure multi-step autonomous reasoning and execution~\cite{liu2024agentbench,debenedetti2024agentdojo,zhou2024webarena,jimenez2024swebench,bhatt2024cyberseceval2,wan2024cyberseceval3}. We instead use formally checkable scope and disclosure constraints to control capability.

Threat analyses for agentic systems emphasize indirect prompt injection, malicious tool outputs, and backdoor behavior~\cite{greshake2023ipi,zhan2024injecagent,yi2025bipia,evertz2025patterns,agentbackdoors2024}. Existing mitigations (prompt filtering, firewalls, and intent classifiers) reduce attack surface but remain heuristic and bypassable in adversarial settings, including under adaptive attacks~\cite{abdelnabi2025firewalls,zhan2025adaptive,nasr2025attackersecond}. This motivates separating semantic mitigation from theorem-backed guarantees at the effect boundary.

On enforcement, prior systems provide DSLs and runtime controls, including AgentSpec, AgentGuard, and FIDES~\cite{agentspec2026,agentguard2025,balunovic2025fides}. Recent work also studies temporal constraints for LLM-agent actions~\cite{kamath2025temporalconstraints}. These systems are valuable engineering baselines. Our contribution is narrower and more formal: finite-trace admissibility semantics, stated mediation and extraction assumptions, and mechanized proofs of monitor soundness, refinement monotonicity, and one-way composition soundness.

Adjacent policy traditions provide important context. In access control, RBAC and ABAC frameworks capture role- and attribute-based authorization structure~\cite{sandhu1996rbac,nist2014abac}, and Cedar develops a modern authorization language with verification-guided engineering~\cite{cutler2024cedar,disselkoen2024builtcedar}. In network security, firewall and network-policy work emphasizes explicit policy objects and analyzable composition~\cite{bartal1999firmato,khurshid2012veriflow,anderson2014netkat}. Complementary lines analyze distributed firewall anomalies and header-space invariant checking~\cite{alshaer2004distributedfirewall,kazemian2012headerspace,kazemian2013realtime}. We build on these lessons, but target adversarial agent effect traces rather than traditional user-subject authorization and packet-forwarding~policies.

The formal basis follows classic monitor-enforcement literature: enforceable policies as safety properties~\cite{schneider2000enforceable,basin2013enforceable}, edit/reference monitors~\cite{ligatti2005edit,erlingsson2000irm,anderson1972}, and compositional reasoning for interacting components~\cite{dealfaro2001interface}. We do not claim to replace this theory. We instantiate it for offensive-capable LLM agents by treating alignment as constrained effect execution, not intent inference, and by making the paper-to-mechanization correspondence explicit.

\section{Conclusion}
\label{sec:conclusion}

Offensive-capable agents create an authorization problem that is easy to state and hard to make precise: they should be powerful inside an approved engagement and constrained everywhere else. Alignment contracts address this problem at the level where enforcement can be made explicit, namely modeled effects that cross a mediation boundary. For this purpose, the framework treats alignment as constrained effect execution, not intent inference.

The formal core gives this boundary a finite-trace semantics. Contracts specify scope, allowed and forbidden descriptors, budgets, and modeled disclosures; admissibility is decidable after event extraction; and monitor-realized traces satisfy the authored contract when the Effect Observability Assumption holds (Theorem~\ref{thm:soundness}). Refinement and one-way composition support modular contract construction. The Lean 4 artifact checks the theorem-labeled claims or records the explicit assumption schemas on which they depend.

The same formalization also clarifies what the approach does not cover. Static pre-admission checks inherit undecidability under the stated reduction assumptions, mediation bypass invalidates the soundness premise, payload and steganographic channels outside the event model are out of scope, and timing channels require temporal predicates not present in the core calculus. More broadly, the framework suggests a design principle for agentic security systems: keep model reasoning outside the trusted core and enforce explicit contracts at the effect boundary. This separation makes the residual risk auditable: operators can inspect which effects are mediated, which disclosures are modeled, and which assumptions sit outside the proof. Future work should quantify retained offensive utility under strict contracts, mechanize controlled contract updates, and evaluate best-effort semantic defenses in deployments.

\bibliographystyle{ACM-Reference-Format}
\bibliography{references-short}

\appendix

\section{Mechanization Tables}
\label{sec:appendix-mechanization-tables}

Table~\ref{tab:notation} summarizes the notation used by the contract semantics in Sections~5--7.

\begin{table}[!htb]
\centering
\small
\begin{tabularx}{\columnwidth}{@{}p{1.8cm}Y@{}}
\toprule
\textbf{Symbol} & \textbf{Meaning} \\
\midrule
$\SigmaEff$ & Effect alphabet (concrete profile here: $\{\Net, \FS, \Exec\}$). \\
$\mathsf{Desc}$ & Effect-descriptor domain $\Sigma_{\sigma \in \SigmaEff}\mathsf{Params}(\sigma)$. \\
$\mathsf{Params}(\sigma)$ & Parameter space for effect class $\sigma$. \\
$\mathsf{Event}$ & Event domain $e=\langle \sigma,p,t,i \rangle$ with $p \in \mathsf{Params}(\sigma)$. \\
$\epsilon,\ \Tr,\ \Tr_{<k}, |\Tr|$ & Empty trace, trace, prefix, and trace length. \\
$\Tr' \sqsubseteq \Tr$ & $\Tr'$ is a prefix of $\Tr$. \\
$\Contract$ & Contract $\langle S, E_{\mathsf{allow}}, E_{\mathsf{forbid}}, B, D, \mathsf{Res}, \mathsf{cost}, \mathsf{flows}\rangle$. \\
$\mathsf{target}(e)$ & Target extracted from event $e$. \\
$S(\cdot)$ & Scope predicate over targets. \\
$E_{\mathsf{allow}}, E_{\mathsf{forbid}}$ & Descriptor predicates $\mathsf{Desc} \to \mathsf{Prop}$, with forbid precedence. \\
$B, \mathsf{cost}_r$ & $B:\mathsf{Resource}\to\mathbb{N}$ and $\mathsf{cost}_r:\mathsf{Event}\to\mathbb{N}$. \\
$D, \mathsf{flows}(e)$ & $D:\mathsf{DataClass}\to\mathsf{Sink}\to\mathsf{Prop}$ and $\mathsf{flows}:\mathsf{Event}\to\mathcal{P}(\mathsf{DataClass}\times\mathsf{Sink})$. \\
$\Adm(e, \Tr, \Contract)$ & Event admissibility. \\
$\Tr \models \Contract$ & Trace satisfaction. \\
\bottomrule
\end{tabularx}
\caption{Notation for the alignment contracts framework.}
\label{tab:notation}
\end{table}

Table~\ref{tab:mechanization-map} gives the canonical symbol-level mapping between paper notation and Lean constructs.

\begin{table*}[!htb]
\centering
\small
\begin{tabularx}{\textwidth}{@{}p{2.4cm}p{3.8cm}Y@{}}
\toprule
\textbf{Paper Symbol} & \textbf{Lean Construct} & \textbf{Notes} \\
\midrule
$\SigmaEff$ & \texttt{Formal.EffectClass} & finite inductive type; web instantiation uses \texttt{net}, \texttt{fs}, \texttt{exec} (mechanization also supports \texttt{chain} as an extension constructor) \\
$e=\langle \sigma,p,t,i \rangle$ & \texttt{Formal.Event} & timestamp and agent id are natural numbers \\
$(\sigma,p)$ descriptor & \texttt{Formal.EffectDesc} & alias of \texttt{EffectClass $\times$ Params}; mechanized core uses a normalized \texttt{Params} carrier; domain-specific schemas compile to this descriptor form \\
$\Tr$ & \texttt{Formal.Trace = List Event} & mechanized core is finite-trace \\
$\Tr' \sqsubseteq \Tr$ & \texttt{Formal.Prefix} & existential append characterization \\
$\Contract$ & \texttt{Formal.Contract} & includes \texttt{inScope}, \texttt{allow}, \texttt{forbid}, \texttt{budget}, \texttt{disclose}, \texttt{resources}, \texttt{cost}, \texttt{flows} \\
$\mathsf{target}(e)$ & \texttt{Formal.target} & extracted from event class and target field \\
$\Adm(e,\Tr,\Contract)$ & \texttt{Formal.Admissible} & conjunction of scope, allow, not-forbid, budget, disclosure \\
$\Tr \models \Contract$ & \texttt{Formal.SatisfiesPrefix} & inductive prefix satisfaction \\
$\Contract' \sqsubseteq \Contract$ & \texttt{Formal.Refines} & includes resource inclusion and equality of cost/flow extractors \\
$\mathsf{compose}(\Contract_1,\Contract_2)$ & \texttt{Formal.compose} & one-way soundness proven under \texttt{ComposeCompat} \\
\bottomrule
\end{tabularx}
\caption{Canonical mapping from paper-level notation to mechanized Lean 4 constructs.}
\label{tab:mechanization-map}
\end{table*}

Table~\ref{tab:theorem-coverage} gives theorem-level coverage, linking each paper claim label to its mechanized Lean theorem and assumption scope.

\begin{table*}[!htb]
\centering
\small
\begin{tabularx}{\textwidth}{@{}p{2.8cm}p{4.8cm}Y@{}}
\toprule
\textbf{Paper Label} & \textbf{Lean Theorem} & \textbf{Mechanized Status} \\
\midrule
\texttt{prop:safety} & \texttt{Formal.safety\_finiteWitness} & Directly mechanized finite bad-prefix witness statement over finite traces. \\
\texttt{lem:prefix-preserve} & \texttt{Formal.prefix\_trans} & Directly mechanized prefix transitivity over finite traces. \\
\texttt{lem:monitor-det} & \texttt{Formal.monitor\_state\_deterministic} & Directly mechanized determinism of monitor run for equal proposed traces. \\
\texttt{thm:soundness} & \texttt{Formal.enforcement\_soundness} & Directly mechanized for monitor-run semantics; deployment transfer uses EOA as an external systems assumption. \\
\texttt{thm:refinement} & \texttt{Formal.refinement\_soundness} & Directly mechanized with explicit \texttt{Refines} obligations. \\
\texttt{thm:composition} & \texttt{Formal.composition\_soundness\_left} & Mechanized one-way implication only (composed $\Rightarrow$ each component), under \texttt{ComposeCompat}. \\
\texttt{prop:decidable} & \texttt{Formal.admissibleDecidable} & Mechanized with finite-list budget/disclosure decidability helpers; complexity bounds are analytical. \\
\texttt{prop:preserving} & \texttt{Formal.preservation\_under\_rules} & Mechanized abstractly for rule sequences preserving \texttt{Safe}. \\
\texttt{lem:eoa} & \texttt{Formal.eoa\_from\_safe} & Auxiliary assumption-transport lemma (if \texttt{Safe $\rightarrow$ EOA}, then per-state implication). \\
\texttt{lem:bridge} & \texttt{Formal.safe\_trace\_bridge} & Auxiliary assumption-transport lemma for the \texttt{Safe}-to-trace-satisfaction bridge. \\
\texttt{thm:compositionality} & \texttt{Formal.adaptation\_compositionality} & Mechanized adaptation assumption schema with explicit \texttt{Safe} and bridge assumptions. \\
\path{thm:concrete-adaptation} & \path{Formal.dep_adaptation_soundness} & Mechanized concrete deployment instance with explicit list-based rule-preservation assumptions. \\
\texttt{thm:impossibility} & \path{Formal.preadmission_undecidable_of_reduction} & Mechanized supplied-reduction transfer theorem; concrete forbidden-event construction is described in the text. \\
\texttt{thm:boundary} & \texttt{Formal.enforcement\_boundary} & Mechanized schema-level equivalence parameterized by forward/backward obligations. \\
\bottomrule
\end{tabularx}
\caption{Theorem-to-mechanization correspondence with explicit assumption scope.}
\label{tab:theorem-coverage}
\end{table*}

\section{Open Science}
The artifact contains the Lean 4 mechanization of the paper's formal core.
\begin{itemize}
\item Lean sources under \path{Formal/}, with entry points \path{Formal.lean} and \path{Proofcheck.lean}.
\item Project configuration in \path{lean-toolchain}, \path{lakefile.lean}, and \path{lake-manifest.json}.
\item Theorem correspondence manifest in \path{artifacts/proof_manifest.json}, plus a successful \texttt{lake build} log.
\end{itemize}
The intended reproducibility check is simply \texttt{lake build} with the pinned Lean toolchain.

\section{Ethical Considerations}
This work studies controls for offensive-capable agentic systems, where misuse risk is highest when deployment occurs without explicit authorization boundaries. The framework therefore centers contract-level scope, allow/forbid constraints, resource budgets, and disclosure policies as enforceable objects. Consistent with the main text, guarantees are conditional on complete mediation (EOA) of $\SigmaEff$-effects and apply only to observable effect behavior; payload/steganographic channels outside the event model, unconstrained timing channels, and payload-semantic intent inference remain outside theorem-level guarantees. The paper does not claim universal prevention of abuse. It claims only what is formalized in the contract semantics and proved under the stated assumptions.

\section{Discussion}
\label{sec:discussion}

\paragraph{Scope of guarantees.}
The guarantees are conditional on the Effect Observability Assumption (EOA): complete mediation of $\SigmaEff$-effects. Payload/steganographic channels outside the event model and timing channels without explicit temporal predicates are outside scope (Theorem~\ref{thm:boundary}). Payload semantics receive best-effort mitigation, not theorem-level guarantees. Contracts also do not repair policy-authoring errors; the monitor enforces the authored contract.

\paragraph{Practical considerations.}
This paper develops the formal core; deployment still requires engineering work: (1) monitor implementations using the tractability basis in Proposition~\ref{prop:decidable}, (2) concrete mediation mechanisms to satisfy EOA (for example kernel hooks and sandboxing), and (3) tooling for contract authoring and review.

\paragraph{Deployment posture.}
The model constrains effect realization, not model intent. For offensive-capable agents, we therefore recommend explicit signed contracts, verifiable mediation, and audit trails as baseline deployment requirements.

\paragraph{Toward dynamic contract updates (future work).}
The mechanized core assumes a static contract (Definition~\ref{def:contract}) during one engagement. Real deployments may require controlled updates, such as scope extension after a verified prerequisite. One possible extension is a witness-gated update request $\langle C_{\mathit{cur}}, \Delta, W, \mathit{ctx} \rangle$, where $W$ is machine-checkable evidence (for example an operator signature or external state proof). Any dynamic update semantics should preserve the main invariants used throughout the paper: decidable admissibility, trace soundness under EOA, and compatibility with refinement and composition definitions (Definition~\ref{def:refinement}, Section~\ref{sec:soundness}). Full formalization of this stateful update model is deferred to later development.

\end{document}